\documentclass[submission]{eptcs}
\providecommand{\event}{ICE 2011}

\usepackage{amsthm}
\newcommand{\mmdef}{\mbox{$\;\stackrel{\textrm{\tiny def}}{=}\;$}}

\newcommand{\st}{\; \mid \;}
  
\usepackage{color,graphics,xcolor}

  \usepackage{ulem}
  \definecolor{shadecolor}{rgb}{1,0.99,0.9}

\newcommand{\mykeyword}[1]{\textsf{\upshape #1}}

\newcommand{\truek}{\mykeyword{true}}

\newcommand{\ENCan}[1]{\langle #1 \rangle}

\newcommand{\participant}[1]{\mathtt{#1}}

\newcommand{\typevar}{\mathbf{t}}

\newcommand{\branch}{\&}

\newcommand{\typeconst}[1]{\mykeyword{#1}}

\newcommand{\End}{\typeconst{end}}

\newcommand{\TO}[2]{\participant{#1}\to\participant{#2}}

\newcommand{\rcdt}{\{l_i\colon T_i\}_{i\in I}}
\newcommand{\brancht}[1][\alpha]{\branch\rcdt}

\newcommand{\grmeq}{\; ::= \;}
\newcommand{\grmor}{\; \mid \;}

\newcommand{\VEC}[1]{\vec{#1}}

\newif\ifny\nytrue
\nyfalse

\newif\ifvv\vvtrue
\vvfalse

\newif\ifkohei\koheitrue
\koheifalse

\newif\ifmarco\marcotrue
\marcofalse

\def\fps@figure{tp}      
\def\fps@table{tp}
\newcommand{\ENTAILS}{\supset}

\newcommand{\DRECA}[4]%
{\mu{\typevar}\textcolor{black}{\ENCan{#2}({#1})\{#4\}}.{#3}}

\newcommand{\GA}{\mathcal{G}}

\newcommand{\ENTAILSb}{\ENTAILS}

\newcommand{\ptp}[1]{{\participant{#1}}}

{\end{array}%
 \right.}

\newcounter{remcounter}
\newtheorem{REM}[remcounter]{Remark}
\newtheorem{definition}[remcounter]{Definition}
\newtheorem{EX}[remcounter]{Example}
\newtheorem{fact}[remcounter]{Fact}
\newtheorem{PRO}[remcounter]{Proposition}
\newtheorem{theorem}[remcounter]{Theorem}

\newcommand{\treeroot}[1]{{#1}^\bullet}

\usepackage{stmaryrd}
\newcommand{\fail}[2]{#1\!\lightning\!_{#2}}

\newcommand{\nodes}{\mathcal{N}}
\newcommand{\trees}{\mathcal{T}}
\newcommand{\variables}{\mathcal{V}}
\newcommand{\predicates}{\Psi}
\newcommand{\participants}{\mathcal{P}}
\newcommand{\HS}{\mathtt{\overline{HS}}}
\newcommand{\HSProblem}{\mathtt{\overline{HS}}}

\newcommand{\responsible}[2]{\mathtt{resp}_{#1}(#2)}
\newcommand{\weaken}[1]{\mathtt{strengthen(#1)}}
\newcommand{\propagate}[1]{\mathtt{propagate(#1)}}

\newcommand{\nada}{\bot}

\newcommand{\methodoneE}{\Phi_1}
\newcommand{\methodtwoE}{\Phi_2}
\newcommand{\methodone}[1]{\Phi_1(#1)}
\newcommand{\methodtwo}[1]{\Phi_2(#1)}

\newcommand{\IOchain}{\;\prec\;}
\newcommand{\IOchainPed}{\prec_{{\scriptscriptstyle T}}}
\newcommand{\prop}[3]{\mathtt{P}_{#1}(#2,#3)}

\newcommand{\TSProblemE}{\mathtt{\overline{TS}}}
\newcommand{\TS}[1]{\mathtt{TS}(#1)}
\newcommand{\TSnodeE}{\mathtt{TSnode}}
\newcommand{\TSnode}[2]{\TSnodeE_#1(#2)}

\newcommand{\TSProblem}[1]{\mathtt{\overline{TS}}(#1)}
\newcommand{\VARS}[1]{\text{{\it var}}(#1)}

\newcommand{\rewriteE}{\mathtt{rewrite}}
\newcommand{\rewrite}[2]{\rewriteE(#1,#2)}
\newcommand{\conflict}{\emph{is in conflict }}
\newcommand{\partitionE}{\mathtt{split}}
\newcommand{\partition}[4]{\partitionE_{#1}\left(#2,#3,#4\right)}
\newcommand{\buildE}{\mathtt{build}}
\newcommand{\build}[3]{\buildE_{#1}\left(#2,#3\right)}
\newcommand{\liftPredE}{\Phi_3}
\newcommand{\liftPred}[1]{\liftPredE(#1)}
\newcommand{\TSresolveE}{\mathtt{TSres}}
\newcommand{\TSresolve}[2]{\TSresolveE_{#1}(#2)}

\newcommand{\carry}[1]{\carryE(#1)}
\newcommand{\carryE}{\mathit{var}}
\newcommand{\carryPed}[2]{\carryE_{#1}(#2)}
\newcommand{\assert}[1]{\assertE(#1)}
\newcommand{\assertE}{\mathit{cst}}
\newcommand{\assertPed}[2]{\assertE_{#1}(#2)}

\newcommand{\parentE}{\mathit{parent}}
\newcommand{\parentPed}[2]{\parentE_{#1}(#2)}

\newcommand{\predE}{\mathtt{PRED}}
\newcommand{\predPed}[2]{\predE_{#1}(#2)}

\newcommand{\spartyE}{\mathit{snd}}
\newcommand{\spartyPed}[2]{\spartyE_{#1}(#2)}
\newcommand{\rpartyPed}[2]{\rpartyE_{#1}(#2)}

\newcommand{\rpartyE}{\mathit{rcv}}

\newcommand{\sr}[2]{\ptp{#1}\;{\rightarrow}\;\ptp{#2}:}
\newcommand{\knows}[2]{\mathtt{knows}_{\ptp{#1}}(#2)}
\newcommand{\varHS}[2]{\mathtt{varHS}_{\ptp{#1}}(#2)}

\newcommand{\uppathPed}[2]{#1 \! \! \uparrow_{#2}}

\newcommand{\values}[4]{\sr{#1}{#2}\{#3 \st #4\}}
\newcommand{\branching}[6]{\sr{#1}{#2}\!\!\left\lgroup\{#6\}#3:#4\right\rgroup_{\!\!\!\!\!#5}\!\!}
\newcommand{\branchingM}[2]{\sr{#1}{#2}}
\newcommand{\recursion}[3]{\mu\ \typevar\ \ENCan{#1}{\{#2 \st #3\}}}
\newcommand{\reccall}[1]{\typevar \ENCan{#1}}
\newcommand{\tree}[1]{\mathtt{T}(#1)}
\newcommand{\assertion}[1]{\mathtt{A}(#1)}
\newcommand{\lab}[1]{\underline{#1}}

\usepackage{pslatex}
\usepackage{latexsym}
\usepackage{graphicx}
\usepackage{amssymb}
\usepackage{amsmath}
\usepackage{color}
\usepackage{latexsym}
\usepackage{graphicx}
\usepackage{url}
\numberwithin{equation}{section}

\usepackage{listings}
\usepackage{pdfsync}
\usepackage{url}

\title{Amending Contracts for Choreographies\thanks{This work has been
    supported by the project Leverhulme Trust Award ”Tracing
    Networks”.  }} \author{Laura Bocchi \qquad Julien Lange \qquad
  Emilio Tuosto \institute{Department of Computer Science, University
    of Leicester, UK}
\email{ bocchi@mcs.le.ac.uk \qquad jlange@mcs.le.ac.uk \qquad emilio@mcs.le.ac.uk}
}
\def\titlerunning{Amending Contracts for Choreographies}
\def\authorrunning{L.\ Bocchi, J.\ Lange and E.\ Tuosto}

\begin{document}

\maketitle
\begin{abstract}
  Distributed interactions can be suitably designed in terms of
  \emph{choreographies}.
  Such abstractions can be thought of as global descriptions of the
  coordination of several distributed parties.
  \emph{Global assertions} define contracts for choreographies by
  annotating multiparty session types with logical formulae to
  validate the content of the exchanged messages.
 %
  The introduction of such constraints is a critical design issue as
  it may be hard to specify contracts that allow each party to be able to
  progress without violating the contract.
  %
  %
  In this paper, we propose three methods that automatically correct
  inconsistent global assertions.
  The methods are compared by discussing their applicability and the
  relationships between the amended global assertions and the original
  (inconsistent) ones.
\end{abstract}

\section{Introduction}\label{sec:intro}

Choreographies are high level models that describe the conversations 
among distributed parties from a global perspective.
\emph{Global types}~\cite{mps} and \emph{global
  assertions}~\cite{bhty10} provide an effective methodology for the
design of distributed choreographies (as e.g.,
in~\cite{carbone.honda.yoshida:esop07}) by allowing static checking of
a number of properties such as deadlock freedom and session fidelity.

Intuitively, global types establish the interaction pattern for the
harmonious coordination of distributed parties while global assertions
combine global types with logic to feature
\emph{design-by-contract}~\cite{scoop}.
Basically, global assertions decorate global types with logical
formulae (\emph{predicates}) that constrain interactions, declaring
senders' obligations and receivers' requirements on exchanged data and on the choice of 
the branches to follow.
This adds fine-grained constraints to the specification of the interaction structure. 
For instance, the global assertion
\begin{equation}\label{eq:ex0}
  \begin{array}{l}
    \values{Alice}{Bob}{a}{a>0}.\\
    \values{Bob}{Carol}{b}{b>a}
  \end{array}
\end{equation}
describes a protocol with three participants, $\ptp{Alice}$,
$\ptp{Bob}$, and $\ptp{Carol}$, who agree on a ``contract''
constraining the \emph{interaction variables} $a$ and $b$.
The contract stipulates that ($i$) $\ptp{Alice}$ has to send a
positive value to $\ptp{Bob}$ in the first interaction, and that
($ii$) $\ptp{Bob}$ is obliged to send $\ptp{Carol}$ a value strictly
greater than the one fixed for $a$ in the first interaction.
Notice that $\ptp{Bob}$ can fulfill his pledge (i.e., the assertion
$b > a$ in the second interaction above) only after he has received
the value $a$ from $\ptp{Alice}$.

Once designed, a global assertion $\GA$ is \emph{projected} on
\emph{endpoint assertions} that are local types -- modelling the behaviour
of a specific participant -- constrained according to the predicates of $\GA$.
For instance, the projection for $\ptp{Alice}$ in the
example~(\ref{eq:ex0}) above is an endpoint assertion prescribing that $\ptp{Alice}$ has to
send a positive value to $\ptp{Bob}$.
Endpoint assertions can be used for static validation of the actual processes
implementing one or more roles in a choreography represented by $\GA$, and/or to synthesise
monitor processes for run-time checking/enforcement.

The methodology described above can be applied only when global
assertions are \emph{well-asserted}~\cite{bhty10}, namely when global
assertions obey two precise design principles:
\emph{history-sensitivity} (HS for short) and \emph{temporal
  satisfiability} (TS for short).
Informally, HS demands that a party having an obligation on a
predicate has enough information for choosing a set of values that
guarantees it.
Instead, TS requires that the values sent in each interaction do not
make predicates of future interactions unsatisfiable.

The main motivation of our interest in HS and TS is that, in
global assertions, they are the technical counterparts of the fundamental
coordination issue that could be summarized in the slogan
``who does what and when does (s)he do it''.
In fact, HS pertains to \emph{when} variables are constrained and
\emph{who} constrains them, while TS pertains to \emph{which} values
variables take.
The contracts specified in global assertions are, on the one hand, ``global''
as they pertain to the whole choreography while, on the other hand, they are
also ``local'' in (at least) two aspects.
The first is that they assign responsibilities to participants
(\emph{who}) at definite moments of the computation (\emph{when}).
The second aspect is that the values assigned to variables are
critical because either one could over-constrain variables fixed in
the past or over-restrict the range of those assigned in the future
(\emph{which}).
These conditions (especially TS) are rather crucial as global assertions
that violate them may be infeasible or fallacious.
For instance, if the predicate for $\ptp{Bob}$ in the second
interaction in~(\ref{eq:ex0}) were $3 > b > a$ then $\ptp{Bob}$ could
not fulfill his contract if $\ptp{Alice}$ had fixed the value $2$ for
$a$ in the first interaction.

Guaranteeing HS and TS is often non-trivial, and this burden is on
the software architect; using tools like the ones described
in~\cite{LT10}, one only highlights the problems but does not help to fix
them.
HS and TS are global semantic properties that may
be hard to achieve.
Namely, TS requires to trace back for ``under-constrained'' interactions (i.e., which allow values causing future predicates to be unsatisfiable)
and re-distribute there the unsatisfiable constraints.

\paragraph{Contributions}
We show a few techniques that help software architects to amend
\emph{global assertions} during the design of distributed
choreographies.
The preliminary notions used in the rest of the paper are given in
\S~\ref{sec:preliminaries}.
In \S~\ref{sec:HS} we give two algorithms which, if applicable,
automatically fix HS in global assertions; the first algorithm
strengthens a predicate while the second one is based on variable
propagation.
In \S~\ref{sec:TS} we give an algorithm which, if possible, moves
predicates up in the global assertion in order to remove TS violations.
%
%
\S~\ref{sec:methodo} outlines a methodology based on the three algorithms.
Conclusions and future work are discussed in \S~\ref{sec:conc}.
%


\section{Preliminaries}
\label{sec:preliminaries}
Let $\participants$ (ranged over by $\ptp{p}, \ptp q,
\ptp{s},\ptp{r},\ldots$) and $\variables$ (ranged over by
$u,v,x,y,\ldots$) be two infinitely countable sets of identifiers.
We assume $\participants \cap \variables = \varnothing$ and call
their elements \emph{participants} and \emph{interaction variables}, respectively.
%
Hereafter, $\VEC{\_}$ represents a list of some elements (for
instance, $\VEC v$ is a list of interaction variables); the
concatenation of $\VEC x$ and $\VEC y$ is denoted by the juxtaposition
$\VEC x \; \VEC y$, and, abusing notation, we confound lists with the
underlying sets of their elements (e.g., $a \in \VEC x$ indicates that
$a$ occurs in the list $\VEC x$).
Also, expressions (ranged over by $e$) include variables in
$\variables$, basic data types (e.g., integers, booleans, etc.), and
usual arithmetic operations/relations; $var(e)$ is the set of (free)
variables in $e$; and, we denote logic implication with the symbol $\ENTAILSb$.

As in~\cite{bhty10}, we parametrise our constructions wrt a logical
language $\Psi$, which we assume to be a decidable fragment of a first-order
logic with expressions and quantifiers on
variables; the set of free interaction variables of $\psi \in \Psi$ is
denoted as $\VARS{\psi}$ and we write $\psi(\VEC v)$ to emphasise that
$\VARS{\psi} \subseteq \VEC v$.

The main ingredients of global assertions are \emph{interactions}, abbreviated $\iota$, which
have the form:
\begin{equation}\label{eq:int}
  \values s r {\VEC{v}}{\psi}
\end{equation}
where $\ptp s, \ptp r \in \participants$ are the \emph{sender} and the
\emph{receiver}, $\VEC v \subseteq \variables$ is a pairwise-distinct
list of variables, and $\psi \in \Psi$.
Variables $\VEC v$ are called \emph{interaction variables} and,
in~(\ref{eq:int}), we say that they are \emph{introduced} by $\ptp s$.
The interaction~(\ref{eq:int}) reads as ``$\ptp{s}$ has to send to $\ptp{r}$ some
values for $\VEC v$ that satisfy $\psi$'' or as ``$\ptp r$
relies that the values fixed by $\ptp s$ for $\VEC v$ satisfy $\psi$''.
For instance,\footnote{For simplicity, we assume the typing of
  variables understood.}
\[
\values s r {v \; w}{ \exists u . v = u \times w }
\]
states that $\ptp{s}$ has the obligation to send $\ptp{r}$ two values
such that the first is a multiple of the second.

\begin{REM}
  In~\cite{bhty10}, interactions specify a channel over which
  participants communicate.
  In~(\ref{eq:int}) we omit channels since they are inconsequential to
  our results (\cite{BettiniCDLDY08LONG} shows that channels can indeed be removed).
\end{REM}
Given $\iota$ as in~(\ref{eq:int}), we define
\[
\spartyE(\iota) \mmdef \ptp{s},
\qquad
\rpartyE(\iota) \mmdef \ptp r,
\qquad
\VARS \iota \mmdef \VEC v,
\qquad \text{and} \qquad
\assertPed{} \iota \mmdef \psi
\]

\medskip

Def.~\ref{def:GA} below is essentially borrowed
from~\cite{bhty10} but for a slightly simplified syntax.
\begin{definition}[Global Assertions]\label{def:GA}
  \emph{Global assertions} are defined by the following productions.
  \[
  \begin{array}{rcl@{\hspace{1cm}}l}
    \GA & \grmeq & \iota .\GA
    & \text{Prefix}
    \\
    & \grmor & \branching s r {l_j}{\GA_j}{j\in J}{\psi_j}
    & \text{Branching}
    \\
    & \grmor & \recursion{\VEC e }{\VEC{{v}}}{\psi}.\GA
    & \text{Recursive definition}
    \\
    & \grmor & \reccall{\VEC{e}}
    & \text{Recursive call}
    \\
    & \grmor & \End   
    & \text{End session}
  \end{array}\]
  where $\psi ,\psi_j \in \predicates$ and
  $l_j$ ranges over a set of labels.
  We let $\GA, \GA',\GA_j$ range over global assertions.
\end{definition}
The first production in Def.~\ref{def:GA} represents an
interaction prefix; interaction variables $\VARS \iota$ are bound in
the continuation of the prefix and in $\assertPed{} \iota$.
The second production allows the selector $\ptp s$ to choose one of
the labels $\{l_j\}_{j\in J}$ and send it to $\ptp{r}$; the choice
of label $l_j$ is guarded by $\psi_j$ (guaranteed by $\ptp{s}$)
and is followed by $G_j$.
The formal parameters $\VEC v \subseteq \variables$ in recursive
definitions\footnote{Variables $\VEC v$ are pairwise distinct and
  their free occurrences in the body of the recursion are bound by the
  recursive definition.} are constrained by the invariant $\psi$ which
must be satisfied at each recursive call (this is guaranteed when the
global assertion satisfies TS).  The initialisation vector $\VEC{e}$
(of the same length as $\VEC v$) specifies the initial values of the
formal parameters.
Recursive calls must be prefix-guarded.

The termination of the session is represented by $\End$
(trailing occurrences are often omitted).
We denote with $\VARS{\GA}$ the set of interaction variables and
recursion parameters in $\GA$.
\begin{REM}
  For simplicity, we assume Barendregt's convention (i.e., bound
  variables are all distinct and they differ from any free variable).
  Moreover, global assertions $\GA$ are \emph{closed}, i.e., each free
  occurrence of $v \in \VARS \GA$ is either preceded by an interaction
  $\iota$ such that $v \in \VARS \iota$ or by a recursive definition
  having $v$ as one of its formal parameters.
\end{REM}
%

A participant $\ptp{p}$ \emph{knows} a variable $v \in var(\GA)$ if either
\begin{itemize} 
\item there is $\iota$ in $\GA$ such that $v \in \VARS \iota$
  and $\ptp p \in \{\spartyE(\iota), \rpartyE(\iota)\}$
\item or there is a recursive definition $\recursion{\VEC{e_1} e
    \VEC{e_2} }{\VEC{v_1} v \VEC{v_2}}{\psi}.\GA'$ in $\GA$ such that
  $\ptp{p}$ knows all the variables\footnote{Assume that the lenght of
    $\VEC e_i$ and $\VEC{e'}_i$ is the same of $\VEC v_i$ for $i \in
    \{1,2\}$.} in $\VARS e$ and, for each recursive invocation
  $\reccall{\VEC{e'_1} e' \VEC{e'_2}}$ in $\GA'$, $\ptp{p}$ knows all
  variables in $\VARS{e'}$. 
\end{itemize}
We denote with $\knows{p}{\GA} \subseteq \VARS{\GA}$ the set of
variables in $\GA$ that $\ptp{p}$ knows.
\begin{EX}\label{ex:1}
Consider the following global assertion
\[\begin{array}{llll}
    \GA_{\text{ex}\ref{ex:1}} &=& \recursion{10}{v}{\psi}.  \\
    && \qquad\values{Alice}{Bob}{v_1}{\psi_1}.\\ 
    &&\qquad \values{Bob}{Carol}{v_2}{\psi_2}. \\
    && \qquad\typevar\ENCan{v_1}
  \end{array}
  \]
  repeatedly executing a computation where ($i$) $\ptp{Alice}$ sends a
  variable $v_1$ to $\ptp{Bob}$ and ($ii$) $\ptp{Bob}$ sends a
  variable $v_2$ to $\ptp{Carol}$.
    At each step, the invariant $\psi$ must be satisfied, namely at
    the first invocation $\psi[10/v]$ must hold and in all subsequent
    invocations $\psi[v_1/v]$ must hold.
\end{EX}
In $\GA_{\text{ex}\ref{ex:1}}$, $\ptp{Alice}$ knows $v_1$, since she
sends it, while $v_1, v_2 \in \knows{Bob}{\GA_{\text{ex}\ref{ex:1}}}$,
since $\ptp{Bob}$ receives $v_1$ and sends $v_2$, respectively.
$\ptp{Carol}$ knows $v_2$, since she receives it.
Also, $v \in \knows{Alice}{\GA_{\text{ex}\ref{ex:1}}} \cap
\knows{Bob}{\GA_{\text{ex}\ref{ex:1}}}$, since $\ptp{Alice}$ and
$\ptp{Bob}$ know $v_1$, the unique variable in the expression of the
recursive call (and they trivially know all the variables in the
initial expression, i.e. the constant $10$).
However, $\ptp{Carol}$ does not know $v$ since she does not know
$v_1$.

 \medskip

It is convenient to treat global assertions as trees whose nodes are
drawn from a set $\nodes$ (ranged over by $n,n',\ldots$) and
labelled with information on the syntactic categories of
Def.~\ref{def:GA}.
Hereafter, we write $n \in T$ if $n$ is a node of a tree $T$,
$\lab{n}$ to denote the label of $n$, and $\treeroot T$ for the root of
$T$.
\begin{definition}[Assertion Tree]\label{def:tree}
  The \emph{assertion tree} $\tree{\GA}$ of a global assertion $\GA$
  is defined as follows:
  \begin{itemize}
  \item If $\GA = \iota.\GA'$ then $\treeroot{\tree{\GA}}$ has label
    $\iota$ and its unique child is $\treeroot{\tree{\GA'}}$.
  \item If $\GA=\branching{s}{r}{l_j}{\GA_j}{j\in J}{\psi_j}$ then
    $\treeroot{\tree{\GA}}$ has label $\TO{s}{r}$ and its children are
    $\{n_j\}_{j \in J} \subseteq \nodes$ such that, for each $j \in
    J$, $\lab{n_j} = \{\psi_j\}l_j$ and $\treeroot{\tree{\GA_j}}$ is
    the unique child of $n_j$.
  \item If $\GA=\recursion{\VEC{e}}{\VEC{v}}{\psi}.\GA'$ then
    $\treeroot{\tree{\GA}}$ has label
    $\recursion{\VEC{e}}{\VEC{v}}{\psi}$ and its unique child is
    $\treeroot{\tree{\GA'}}$.
  \item If $\GA=\typevar{\ENCan{\VEC{e}}}$ then $\tree{\GA}$ consists
    of one node with label $\typevar{\ENCan{\VEC{e}}}$.
  \item If $\GA=\End$ then $\tree{\GA}$ consists of one node with
    label $\End$.
  \end{itemize}
  We denote the set of assertion trees as $\trees$ and let
  $T,T',\ldots$ range over $\trees$.
\end{definition}
For convenience, given $T\in \trees$, we will use the partial functions
\[
\carryE_T:\nodes\rightarrow2^\variables, \qquad 
\assertE_T:\nodes\rightarrow\predicates, \qquad \text{and} \qquad 
snd_T,rcv_T:\nodes\rightarrow\participants
\]
that are undefined\footnote{We write $f(x) = \nada$ when the
  function $f$ is undefined on $x$.} on $\nodes \setminus \{n \st n
\in T \}$ and defined as follows otherwise:
%
%
\[\begin{array}{l@{\hspace{1.5cm}}l}\small
\carryPed{T}{n} =
  \begin{cases}
    \VARS{\iota}, & \text{if } \lab{n} = \iota \\
    \emptyset,  & \text{otherwise} 
  \end{cases}
&
\assertPed{T}{n} =
  \begin{cases}
    \psi, & \text{if } \lab{n} = \iota \text{ and } \assertPed{} \iota = \psi, \text{ or }  \lab{n}=\{\psi\}l\\
    \truek, & \text{otherwise}
  \end{cases}
\\[2em]
\spartyPed{T}{n} =
  \begin{cases}
    \spartyE(\iota), & \text{if }\lab{n}=\iota \\
    \ptp{s}, & \text{if }\lab{n}={\ptp{s}}\;\rightarrow\;{\ptp{r}} \\
  \end{cases}
&
\rpartyPed{T}{n} =
  \begin{cases}
    \rpartyE(\iota), & \text{if }\lab{n}=\iota\\
    \ptp{r}, & \text{if }\lab{n}={\ptp{s}}\;\rightarrow\;{\ptp{r}} \\ 
  \end{cases}
\end{array}\]
Moreover, we shall use the following functions:
\begin{itemize}
\item $\parentPed{T}{n}$ returning $\epsilon$ if $n = \treeroot T$,
  the parent of $n$ in $T$ if $n \in T$, and $\nada$ otherwise.
\item $\uppathPed{n}{T}$ returning the path from $\treeroot T$ to $n$
  if $n \in T$, and $\nada$ otherwise.
\end{itemize}

Given $T \in \trees$, let $\assertion T$ be the global assertion
obtained by appending the labels of the nodes in (depth-first)
preorder traversal visit of $T$.
\begin{fact}\label{fact:text}
  $\assertion{\tree{\GA}}=\GA$
\end{fact}
Fact~\ref{fact:text} allows us to extend $\knows{p}{\_}$ to $\trees$
by $\knows{p}{T} \mmdef \knows p {\assertion T}$.

\begin{fact}\label{fact:assert}
If $T\in\trees$ then $\tree{\assertion{T}}=T$
\end{fact}
Facts~\ref{fact:text} and~\ref{fact:assert} basically induce an
isomorphism between global assertions and their parsing trees.



\section{Towards a Better Past}
\label{sec:HS}
In a distributed choreography, parties have to make local choices on
the communicated values; such choices impact on the graceful coordination
of the distributed parties.
It is therefore crucial that the responsible party has ``enough
information'' to commit to an ``appropriate'' local choice, in each
point of the choreography.
For global assertions, this distills into \emph{history sensitivity} (HS), a
property defined in~\cite{bhty10} demanding each sender/selector
to know all the variables involved in the predicates (s)he must guarantee.
We illustrate HS with Example~\ref{ex:hs} below.
\begin{EX}\label{ex:hs}
  The global assertion $\GA_{\text{ex}\ref{ex:hs}}$ violates HS.
\small  \[\begin{array}{llll}
    \GA_{\text{ex}\ref{ex:hs}} &=& \values{Alice}{Bob}{v_1}{v_1>0}.\\ 
    && \values{Bob}{Carol}{v_2}{v_2>0}. \\
    && \values{Carol}{Alice}{v_3}{v_3>v_1}
  \end{array}\]
 \normalsize In fact, $\ptp{Carol}$'s obligation
  $v_3>v_1$ cannot be fulfilled because  $v_1 \not\in \knows{Carol}{\GA_{\text{ex}\ref{ex:hs}}}$.
\end{EX}


Given a global assertion $\GA$, the function $\HS(\GA)$ below returns the
nodes of $\tree{\GA}$ where HS is violated
\[
 \HS(\GA) \mmdef
 \{ n \in \tree{\GA} \st var(\assertPed{T}{n})\not\subseteq \knows{s}{\uppathPed{n}{T}} \text{ and } \ptp{s}=\responsible{\tree{\GA}}{n} \}
\]
where $\responsible{T}{\_} : \nodes \to \participants$ yields the
responsible party of a node and is defined as
\[\responsible{T}{n} \mmdef
  \begin{cases}
    \spartyPed{T}{n}, & \text{if } \lab{n}=\iota \\
    \spartyPed{T}{\parentPed{T}{n}}, & \text{if }\lab{n}=\{\psi\}l \\
    \nada, & \text{otherwise} 
  \end{cases}
\]
Intuitively, to determine whether a node $n \in \tree{\GA}$ violates
HS, one checks if the responsible party of $n$ knows all the variables involved
in $\assertPed{\tree{\GA}}{n}$.

%

Given $T \in \trees$, $\varHS{T}{\_}: \nodes \to 2^\variables$ 
is defined as
\[
\varHS{T}{n} \mmdef var(\assertPed{T}{n})\setminus
\knows{s}{\uppathPed{n}{T}} \;\;\; \text{where}\;\;\;
\ptp{s}=\responsible{T}{n}
\]
Namely, $\varHS{T}{n}$ yields the variables of $n$ not known to the
responsible party of $n$.
It is a simple observation that if HS is violated in a node $n$, then
there exists a variable in the predicate of $n$ which is not known to
the responsible party of $n$ (namely
if  $n \in \HS(\GA)$ then $\varHS{T}{n}  \not =\varnothing$).
\begin{EX}\label{ex:3}
  Consider the following global assertion:
  \small\[\begin{array}{llll}
    \GA_{\text{ex}\ref{ex:3}} &=& \recursion{10}{v}{v>0}.\\
    && \qquad \values{Alice}{Bob}{v_1}{v \geq v_1}.\\ 
    && \qquad \values{Bob}{Carol}{v_2}{v_2>v_1}. \\
    && \qquad \values{Carol}{Alice}{v_3}{v_3>v_1}.\\
    && \qquad \values{Carol}{Bob}{v_4}{v_4>v}.\\
    && \qquad \typevar\ENCan{v_1}
  \end{array}\]
 \normalsize %
  $\HS(\GA_{\text{ex}\ref{ex:3}})=\{n_3,n_4\}$ where $n_3$ and $n_4$ are the
  nodes in $\tree{\GA_{\text{ex}\ref{ex:3}}}$ corresponding to the
  third and fourth interactions of $\GA_{\text{ex}\ref{ex:3}}$, i.e.\
  $\lab{n_3} = \values{Carol}{Alice}{v_3}{v_3>v_1}$ and
  $\lab{n_4} =\values{Carol}{Bob}{v_4}{v_4>v}$.
\end{EX}
In Example~\ref{ex:3}, $\ptp{Carol}$ is responsible for both
violations (i.e.,
$\responsible{\tree{\GA_{\text{ex}\ref{ex:3}}}}{n_3}=\responsible{\tree{\GA_{\text{ex}\ref{ex:3}}}}{n_4}=\ptp{Carol}$).
$\varHS{\tree{\GA_{\text{ex}\ref{ex:3}}}}{n_3}=\{v_1\}$ (i.e.,
$\ptp{Carol}$ has an obligation on $v_3>v_1$ without knowing $v_1$)
and the violation in $n_4$ is on
$\varHS{\tree{\GA_{\text{ex}\ref{ex:3}}}}{n_4}=\{v\}$ (i.e.,
$\ptp{Carol}$ has an obligation on $v_4>v$ without knowing $v$).
Note that the violation on HS does not imply that $\ptp{Carol}$ will
actually violate the condition $v_3>v_1$.
In fact, $\ptp{Carol}$ could unknowingly choose either a violating or
a non violating value for $v_3$.

In \S~\ref{sec:weakening} and \S~\ref{sec:varprop}, we present two
algorithms that fix, when possible, violations of HS
in a global assertion.
We discuss and compare their applicability, as well as the
relationship between the amended global assertion and the original
one.
We shall use Example~\ref{ex:3} as the running example of
\S~\ref{sec:weakening} and \S~\ref{sec:varprop}.

\subsection{Strengthening}\label{sec:weakening}

Fix a global assertion $\GA$ and its assertion tree $T=\tree{\GA}$.
Assume HS is violated at $n \in T$ and $\assertPed{T}{n}=\psi$.
Violations occur when the responsible party $\ptp{s}$ of $n$
is ignorant of at least one variable $v \in var(\psi)$.
The strengthening algorithm (cf. Def.~\ref{def:hs1}) replaces $\psi$ in
$\GA$ with an assertion $\psi[v'/v]$ so that
\begin{enumerate}
\item[(1)] $v'$ is a variable that $\ptp{s}$ knows, 
\item[(2)] if $\psi[v'/v]$ and the predicates occurring from $\treeroot T$ to $\parentPed{T}{n}$ are satisfied then also $\psi$ is satisfied.
\end{enumerate}
If there is no variable $v'$ that ensures (1) and (2) then we say that
\emph{strengthening is not applicable}.
Intuitively, the method above \emph{strengthens} $\psi$ with $\psi[v'/v]$.
Due to (2), $\psi$ can be still guaranteed relying on
the information provided by all the predicates occurring before $n$.
%
Let $\predE_T : \nodes \rightarrow \predicates$ yield the conjunction
of the predicates on the path from $\treeroot T$ to the parent of a node:
\[
\predPed{T}{n} \mmdef
\begin{cases}
  \nada,
  & \text{if } \parentPed{T}{n} = \nada
  \\
  \truek,
  & \text{if } \parentPed{T}{n} = \epsilon
  \\
  \assertPed{T}{\parentPed{T}{n}}\;\land\; \predPed{T}{\parentPed{T}{n}},
  & \text{otherwise}
\end{cases}
\]

The function $\weaken{\GA}$ uses $\predE_T$ to compute a global
assertion $\GA'$ by replacing in $\GA$, if possible, the assertion
violating HS with a stronger predicate.

\begin{definition}[$\mathtt{strengthen}$]\label{def:weak} If
  $\HS(\GA)=\varnothing$ then $\weaken{\GA}$ returns $\GA$.
  If $n \in \HS(\GA)$, $v\in\varHS{T}{n}$ and there exists $v'
  \in\knows{s}{\uppathPed{n}{T}}$ such that

\begin{equation}\label{eq:hs1l}
\predPed{T}{n} \land \psi[v'/v]\ENTAILSb \psi \quad \text{ with } \quad \psi=\assertPed{T}{n}
\end{equation}
then $\weaken{\GA}$ returns $\assertion{T'}$ where $T'$ is obtained
from $T$ by replacing $\psi$ with $\psi[v'/v]$ in $\lab{n}$.

Finally, when the two cases above cannot be applied, $\weaken{\GA}$
returns $\fail{\GA}{n}$, namely it indicates that $\GA$ violates HS at
$n \in \HS(\GA)$.
\end{definition}

The algorithm $\Phi_1$ in Def.~\ref{def:hs1} recursively applies
$\weaken{\_}$ until either the global assertion satisfies HS or
$\Phi_1$ is not applicable anymore.
\begin{definition}[$\Phi_1$]\label{def:hs1}
  The algorithm $\Phi_{1}$ is defined as follows
\[\methodone{\GA} \mmdef \begin{cases}
   \weaken{\GA}, & \text{if }\;\;  \weaken{\GA}\in \{ \GA ,   \fail{\GA}{n} \}\\
   \methodone{\weaken{\GA}}, &  \text{otherwise}
  \end{cases}\]
\end{definition}

\begin{EX}\label{ex:4}
  Consider $\GA_{\text{ex}\ref{ex:3}}$ from Example~\ref{ex:3} and recall that
  $\HS(\GA_{\text{ex}\ref{ex:3}})=\{n_3,n_4\}$.
  Strengthening is applicable to $n_3$ where we can substitute $v_1$
  with $v_2$ in $v_3>v_1$ to satisfy condition (\ref{eq:hs1l}) in
  Def.~\ref{def:weak}:
\[(v>0\land v \geq v_1 \land v_2>v_1) \land (v_3>v_2) \ENTAILSb (v_3>v_1)\]
The invocation of $\weaken{\GA_{\text{ex}\ref{ex:3}}}$ returns
\small\[\begin{array}{llll}
    \GA' &=& \recursion{10}{v}{v>0}.\\
    && \qquad \values{Alice}{Bob}{v_1}{v \geq v_1}.\\ 
    && \qquad \values{Bob}{Carol}{v_2}{v_2>v_1}. \\
    && \qquad \values{Carol}{Alice}{v_3}{v_3>v_2}.\\
    && \qquad \values{Carol}{Bob}{v_4}{v_4>v}.\\
    && \qquad \typevar\ENCan{v_1}
  \end{array}\]
 \normalsize %
  The invocation of $\weaken{\GA'}$ returns $\fail{\GA'}{n_4}$ since
  $\GA'$ has still one violating node $n_4$ for which strengthening is not
  applicable e.g., $(v>0\land v \geq v_1 \land v_2>v_1 \land v_3>v_2) \land
  (v_4> v_2) \not \ENTAILSb (v_4>v)$.
\end{EX}

\newcommand{\chain}[2]{\mathtt{chain}_{#1}{(#2)}}

\subsection{Variable Propagation}\label{sec:varprop}

An alternative approach to solve HS problems is based on the
modification of global assertions by letting responsible parties of the violating nodes
know the variables causing the violation.
The idea is that such variables are propagated within a ``chain of
interactions''.
\begin{definition}[$\IOchainPed$]\label{def:chain}
  Let $n,n'\in T$, $n \IOchainPed n'$ iff ${n}$ appears in
  $\uppathPed{n'}{T}$ and $ \rpartyPed{T}{n}=\spartyPed{T}{n'}$.
  A vector of nodes ${n_1},\ldots,n_t$ is a \emph{chain in
    $T$} iff $n_i \IOchainPed n_{i+1}$ for all $i\in\{1,\ldots,t-1\}$.
\end{definition}
The relation $\IOchainPed$ is similar to the IO-dependency defined in
\cite{mps} but does not consider branching, since a branching does not
carry interaction variables.

Fix a global assertion $\GA$; let $T = \tree{\GA}$, $n\in \HS(\GA)$,
$v \in \varHS{T}{n}$, and $\ptp{s}=\responsible{T}{n}$.

The \emph{propagation} algorithm (cf. Def.~\ref{def:hs5}) is
\emph{applicable} only if there exists a $\IOchainPed$-chain in
$\uppathPed{n}{T}$ through which $v$ can be propagated from a node
whose sender knows $v$ to $n$, in which $\ptp{s}=\responsible{T}{n}$
can receive it.
Given a chain $\VEC n = n_1 \cdots n_t$ in $T$, let the
\emph{propagation of $v$ in $\VEC n$} be the tree $T' \in \trees$
obtained by updating the nodes in $T$ as follows:
\begin{itemize}
\item $\carryPed{T'}{n_1} = \carryPed{T}{n_1} \cup \{v_1\}$ and
  $\assertPed{T'}{n_1} = \assertPed{T}{n_1}\land (v_1 = v)$,
  with $v_1\in\variables$ fresh.
\item for $i=2\ldots t-1$,
   $\carryPed{T'}{n_i} = \carryPed{T}{n_i} \cup \{v_i\}$ and
   $\assertPed{T'}{n_i} =  \assertPed{T}{n_i}\land (v_{i}=v_{i-1})$,
   with $v_2,\ldots,v_{t-1} \in \variables$ fresh.
\item $\assertPed{T'}{n_t}=\assertPed{T}{n_t}[v_{t-1}/v]$
\item all the other nodes of $T$ remain unchanged.
\end{itemize}
For a sequence of nodes $\VEC n$, $\prop{T}{v}{\VEC{n}}$ denotes $T'$
as computed above if $\VEC n$ is a $\IOchainPed$-chain and $\bot$
otherwise.

\begin{EX}\label{ex:4b}
  In the global assertion $\GA_{\text{ex}\ref{ex:4b}}$ below assume
  $\ptp{Alice}$ knows $v$ from previous interactions (the ellipsis in
  $\GA_{\text{ex}\ref{ex:4b}}$).
 \small \[\begin{array}{llll}
  \GA_{\text{ex}\ref{ex:4b}} &=& \ldots &  \values{Alice}{Bob}{u_1}{\psi_1}.\\ 
  & & &   \values{Bob}{Carol}{u_2}{\psi_2}. \\
  & & & \values{Bob}{Dave}{u_3}{\psi_3}. \\
  & & & \values{Dave}{Alice}{u_4}{u_4>v}
  \end{array}\]
\normalsize
%
%
  For the chain $\VEC{n} = n_1 \; n_3 \; n_4$ in $\tree{\GA_{\text{ex}\ref{ex:4b}}}$ (where
  $n_i$ corresponds to the $i$-th interaction in $\GA_{\text{ex}\ref{ex:4b}}$),
  $\prop{\tree{\GA_{\text{ex}\ref{ex:4b}}}}{v}{\VEC{n}}$ returns $T'$ such that
  $\assertion{T'}$ is simply $\GA_{\text{ex}\ref{ex:4b}}$ with $\psi_1$ replaced by
  $\psi_1\land v =v_1$, $\psi_3$ replaced by $\psi_3\land v_1=v_2$,
  and $\psi_4$ replaced by $u_4 > v_2$
  and the fresh variables $v_1$ and $v_2$ is added to the interaction variables
  of the first and third interactions, respectively.
\end{EX}

We define a function $\mathtt{propagate}$ which takes a global
assertion $\GA$ and returns: (1) $\GA$ itself if HS is satisfied, (2)
$\fail{\GA}{n}$ if HS is violated at $n\in\tree{\GA}$ and propagation
is not applicable, (3) $\GA'$ otherwise, where $\GA'$ is obtained by
propagating a violating variable $v$ of node $n$; in the latter case,
observe that $v$ has been surely introduced in a node $n'\in
\uppathPed{n}{\tree{\GA}}$ from which $v$ can be propagated, since we
assume $\GA$ closed.

\begin{definition}[$\mathtt{propagate}$] The function $\propagate{\GA}$ returns
\begin{itemize}
\item $\GA$, if $\HS(\GA)=\varnothing$
\item $\prop{T}{v}{\VEC{n}}$, if $T=\tree{\GA}$ and there exists
  $n\in\HS(\GA)$ with $v\in\varHS{T}{n}$ and there exists $\VEC{n}=n_0
  \; \VEC{n_1} \; n$ chain in $T$ such that $\spartyPed{T}{n_0}$ knows
  $v$
\item $\fail{\GA}{n}$ with $n\in\HS(\GA)$ otherwise.
\end{itemize}
\end{definition}

\begin{EX}\label{ex:5}
  Consider again the global assertion $\GA'$ obtained after the
  invocation $\weaken{\GA_{\text{ex}\ref{ex:3}}}$ in
  Example~\ref{ex:4}.
  In this case $\HS(\GA')=\{n_4\}$ with $n_4 =
  \values{Carol}{Bob}{v_4}{v_4>v}$.
  Propagation is applicable to $n_4$ and $\propagate{\GA'}$ returns
\small\[\begin{array}{llll}
    \GA'' &=& \recursion{10}{v}{v>0}.\\
    && \values{Alice}{Bob}{v_1}{v \geq v_1}.\\ 
    && \values{Bob}{Carol}{v_2 \; u_1}{v_2>v_1\land u_1=v}. \\
    && \values{Carol}{Alice}{v_3}{v_3>v_2}.\\
    && \values{Carol}{Bob}{v_4}{v_4>u_1}.\\
    && \typevar\ENCan{v_1}
  \end{array}\]
  \normalsize
  by propagating $v$ from the second interaction where the sender
  $\ptp{Bob}$ knows $v$ to $\ptp{Carol}$, $\GA''$ satisfies HS.
  The predicate of the last interaction derives from the substitution
  $(v_4>v)[u_1/v]$.
\end{EX}

The propagation algorithm is defined below and is based on a repeated
application of $\propagate{\_}$.

\begin{definition}[$\Phi_{2}$]\label{def:hs5}
  Given a global assertion $\GA$, the function $\Phi_2$ is defined as
  follows:
\[\methodtwo{\GA}= \left\{
  \begin{array}{lll}
   \propagate{\GA}, & \text{if }\;\;  \propagate{\GA}\in\{\GA,\fail{\GA}{n}\}\\
   \methodtwo{\propagate{\GA}}, &  \text{otherwise}
  \end{array}
\right.\]
\end{definition}

\subsection{Properties of $\Phi_1$ and $\Phi_2$}
We now discuss the properties of the global assertions amended by each
algorithm and we compare them.
Hereafter, we say $\Phi_1$ (resp. $\Phi_2$) returns $\GA$ if either it
returns $\GA$ or it returns $\fail{\GA}{n}$ for some $n$.

The applicability of $\Phi_1$ depends on whether it is possible to
find a variable known by the responsible party of the violating node
such that condition (\ref{eq:hs1l}) in Def.~\ref{def:weak} is
satisfied.
The applicability of $\Phi_2$ depends on whether there exists a chain
through which the problematic variable can be propagated.\footnote{
  Linearity of the underlying multiparty session types (i.e., a
  property that ensures the existence of a dependency chain between
  the interactions)~\cite{mps} does not guarantee that $\Phi_2$ is
  always applicable. The reason is that $n_1 \IOchain n_2$ in the
  sense of \cite{mps} does not imply $n_1 \IOchainPed n_2$ since
  $\IOchainPed$ does not take into account branching but only
  interactions.}

Notably, there are cases in which $\Phi_1$ is applicable and $\Phi_2$
is not, and vice versa.
Also, $\Phi_1$ and $\Phi_2$ return, respectively, two different global assertions from
the original one; hence it may not always be clear which one should be preferred.

\begin{REM}\label{rmk:secrecy}
  In distributed applications it is often necessary to guarantee that
  exchanged information is accessible only to intended participants.
  %
  It is worth observing that $\Phi_2$ discloses information about the
  propagated variable to the participants involved in the propagation
  chain.
  The architect should therefore evaluate when it is appropriate to
  use $\Phi_2$.
\end{REM}

First we show that both $\Phi_1$ and $\Phi_2$ do not change the
structure of the given global assertion.
\begin{PRO}
  Let $\GA$ be a global assertion.
  If $\Phi_1(\GA)$ or $\Phi_2(\GA)$ return $\GA'$ then $\tree{\GA}$
  and $\tree{\GA'}$ are isomorphic, namely they have the same tree
  structure, but different labels.
\end{PRO}

Whereas $\Phi_1$ does not change the underlying type of the global
assertion, $\Phi_2$ does. Indeed, in the resulting global
assertion, more variables are exchanged in each interaction involved
in the propagation. However, the structure of the tree remains the same.

Let $erase(\GA)$ \label{def:erase} be the function that returns the underlying
global type~\cite{mps} corresponding to $\GA$ (i.e.\ a global
assertion without predicates).
\begin{PRO}[Underlying Type Structure]\label{prop:str}
  Let $\GA$ be a global assertion,
  \begin{itemize}
  \item if $\methodone{\GA}$ returns $\GA'$ then
    $erase(\GA)=erase(\GA')$ 
  \item if $\methodtwo{\GA}$ returns $\GA'$ then for all $n\in
    \tree{\GA}$ and its corresponding node $n'\in
    \tree{\GA'}$, \[\carryPed{\tree{\GA}}{n}\subseteq\carryPed{\tree{\GA'}}{n'}\]
  \end{itemize}
\end{PRO}
\begin{proof}[Proof sketch]
  The proof is by induction on the structure of $\GA$ and it trivially follows from
  the fact that neither $\methodoneE$ nor $\methodtwoE$ changes the structure of the assertion tree.
  In fact, $\methodoneE$ changes only the predicates. On the other hand, $\methodtwoE$ changes the predicates
  and {adds} fresh variables to interaction nodes, therefore changing the type of
the exchanged data.
\qedhere
\end{proof}

 The application of $\Phi_1$ and $\Phi_2$ affects the predicates
  of the original global assertion. In $\Phi_1$, strengthening allows less 
  values for the interaction variables of the amended interaction.
Conversely, the predicates computed by $\Phi_2$ are equivalent to the
original ones (i.e., they allow sender and receiver to chose/expect
the same set of values). Nevertheless, such predicates are
syntactically different as $\Phi_2$ adds the equality predicates on
the propagated variables.

%

\begin{PRO}[Assertion Predicates]\label{prop:ap} Let $\GA$ be a global assertion,
  \begin{enumerate}
  \item \label{p:1} if $\methodone{\GA}$ returns $\GA'$ then for all $n\in
    \tree{\GA}$ whose label is modified by $\methodoneE$
    and its corresponding node $n'\in \tree{\GA'}$
    (cf. Proposition~\ref{prop:str}), it holds that
    $\predPed{\tree{\GA'}}{n'}\land \assertPed{\tree{\GA'}}{n'}
      \ENTAILSb \assertPed{\tree{\GA}}{n}$

  \item \label{p:2}  if $\methodtwo{\GA}$ returns $\GA'$ then for all $n\in \tree{\GA}$
    whose label is modified by $\methodtwoE$ and its corresponding node $n'\in \tree{\GA'}$
    \begin{enumerate}
    \item \label{p:a} $\assertPed{\tree{\GA'}}{n'}$ is the predicate
      $\assertPed{\tree{\GA}}{n} \land \psi$
    \item \label{p:b}$\predPed{\tree{\GA}}{n} \ENTAILSb
      \assertPed{\tree{\GA}}{n} \land \psi \iff
      \predPed{\tree{\GA'}}{n'}\ENTAILSb\assertPed{\tree{\GA'}}{n'}$
    \end{enumerate}
    For some $\psi \in \Psi$ satisfiable.
  \end{enumerate}
\end{PRO}

\begin{proof}[Proof sketch]
The proof of item~\ref{p:1} relies on the fact that $\methodoneE$ either does not
change $\GA$ or replaces a problematic variable by a variables for which~\eqref{eq:hs1l} holds.
The proof of item~\ref{p:2} relies on Def.\ref{def:chain}, i.e.\ a predicate of the form $v_1 = v$ or $v_i = v_{i-1}$ is added to each predicate of the nodes in the chain. The additional predicates are satisfiable since they constrain only fresh variables (i.e.\ $v_i$).
\qedhere
\end{proof}
The statement \ref{p:b} in Proposition~\ref{prop:ap} amounts to say
  that $\assertPed{\tree{\GA}}{n} \land \psi$ is equivalent to
  $\assertPed{\tree{\GA'}}{n'}$ when such predicates are taken in
  their respective contexts.

Finally, we show that $\Phi_1$ and $\Phi_2$ do not add violations (of either HS or TS) to the amended global assertions (Proposition~\ref{pro:preserve12}) and that if the return value is not of the type $\fail{\GA}{n}$ then 
the amended global assertion satisfies HS (Theorem~\ref{pro:improve12}). 

\begin{PRO}[Properties Preservation]\label{pro:preserve12}
  Assume $\Phi_{i}(\GA)$ returns $\GA'$ with $i\in\{1,2\}$.
  If $\HSProblem(\GA) = \varnothing$ then $\HSProblem(\GA') = \varnothing $
  and if $\TSProblem{\GA} = \varnothing$ then $\TSProblem{\GA'} = \varnothing$.
\end{PRO}

\begin{proof}[Proof sketch]
The proof of HS preservation by both algorithms follows by the fact that
they both return $\GA$ if $\HSProblem(\GA) = \varnothing$.
TS preservation in $\methodoneE$ follows from the fact that predicates
may only be changed by a variable substitution.
For $T = \tree{\GA}$, such that $\TSProblem{\GA} = \varnothing$,
we have that, for any $n \in T$
\[
\predPed{T}{n} \ENTAILSb \exists \carryPed{T}{n} . \phi
\]
by definition of TS. And, by~\eqref{eq:hs1l}, we have that
\[
\predPed{T}{n} \ENTAILSb \exists \carryPed{T}{n} . \phi[v/v']
\]
i.e.\ TS is preserved by $\methodoneE$ .
TS preservation in $\methodtwoE$ follows from the fact that the predicates of a global assertions are only modified by adding equalities between problematic variables and fresh variables (see statement~\ref{p:b} in Propostion~\ref{prop:ap}).
\qedhere
\end{proof}

\begin{theorem}[Correctness]\label{pro:improve12}
  If there is $\GA'$ such that $\Phi_1(\GA)={\GA'}$ or
  $\Phi_2(\GA)={\GA'}$ then $\HSProblem(\GA') = \varnothing$.
\end{theorem}

\begin{proof}[Proof sketch]
We only consider the cases where the algorithms do return a different tree.
The proof for $\methodoneE$ follows simply from the fact that, at each iteration of the
algorithm, the variable chosen to replace the problematic one is selected so that the
responsible party \emph{knows} it.

The proof for $\methodtwoE$ is by induction on the length of
the $\IOchainPed$-chain at each iteration, and follows from the condition to form such a
chain.
Let $T$ be an assertion tree, $\VEC n = n_1 \ldots n_t$ be the $\IOchainPed$-chain used to solve a 
HS problem at $n \in T$ on a variable $v$.
By construction, the sender of $n_1$ knows $v$, and each variable $v_i$ added at $n_i$ is known to the sender of $n_i$ (by definition of $\mathtt{knows}$).
In addition, the receiver of the $n_t$ is the \emph{responsible party} of $n$, who therefore
knows the variable $v_t$ which replaces $v$ in $n$. 
\qedhere
\end{proof}



\newcommand{\gsatE}{\textit{GSat}}
\newcommand{\gsat}[2]{\gsatE(#1,#2)}

\section{Back to the Future}\label{sec:TS}
In a distributed choreography, the local choices made by some parties
may restrict later choices of other parties to the point that no
suitable values is available.
This would lead to an abnormal termination since the choreography
cannot continue.
For global assertions, this distills into \emph{temporal
  satisfiability} (TS) which requires that the values sent in each
interaction do not compromise the satisfiability of future
interactions.
The formal definition of temporal satisfiability is adapted
from~\cite{bhty10}.
\begin{definition}[TS~\cite{bhty10}]\label{def:ts}
  A global assertion $\GA$ \emph{satisfies TS} (in symbols $\TS{\GA}$) iff
  $\gsat{\GA}{\truek}$ holds where
  \[
  \gsat{\GA}{\psi} \textit{iff} {\small\begin{cases}
      \gsat{\GA'}{\psi \land \assertPed{} \iota}, & \text{if}\ \GA =
      \iota.\GA' \text{ and } \psi \ENTAILSb \exists \VARS
      \iota. \assertPed{} \iota
      \\[1pc]
      \displaystyle\bigwedge_{j \in J} \gsat{\GA_j}{\psi \land \psi_j}, &
      \text{if}\ \GA = \branching s r {l_j}{\GA_j}{j\in J}{\psi_j} \text{
        and } \psi \ENTAILSb \displaystyle\bigvee_{j \in J} \left( \psi_j
      \right)
      \\[1pc]
     \gsat{\GA'}{\psi \land \psi'}, & \text{if}\ \GA = \recursion{\VEC e
      }{\VEC{{v}}}{\psi'}.\GA' \text{ or } \GA=\typevar_{\psi'(\VEC v)}\ENCan{\VEC{e}}, \text{ and } \psi \ENTAILSb \psi'[\VEC e / \VEC v]
      \\[1pc]
      \GA =   \End, & \text{otherwise}
    \end{cases}}
  \]
  For an assertion tree $T \in \trees$, $\TS T$ holds iff
  $\gsat{\assertion T}\truek$.
\end{definition}
Intuitively, $\psi$ in $\gsatE$ is equivalent to the conjunction of
all the predicates that precede an interaction.
In the first case, all the values satisfying $\psi$ allow to
instantiate the interaction variables $\carry{\iota}$ so to satisfy
the constraint $\assertPed{}\iota$ of $\iota$.
For branching, $\gsatE$ requires that at least one branch can be
chosen and that each possible path satisfies $\gsatE$.
The recursive definition requires that the initial parameters satisfy the invariant
$\psi'$. In recursive calls, we assume an annotation giving the invariant of the corresponding
recursive definition (i.e.\ $\psi'(\VEC v)$).

Often, TS problems appear when one tries to restrict the domain of a variable after its introduction.
To illustrate this, we introduce the following running example.
\begin{EX}\label{ex:tsex}
  Consider $\GA_{\text{ex}\ref{ex:tsex}}$ below, where $\ptp p$
  constraints $x$ and $y$: {\small
  \[\begin{array}{lll}
    \GA_{\text{ex}\ref{ex:tsex}} & = & \values p q x {x < 10}. \\
    && \values p q y {y > 8}. \\
    &&  \values q p z {x > z  \land\  z> 6 \ \land\  y \neq z}
  \end{array}\]
}
  When $\ptp q$ introduces $z$, both 
$x$ and $y$ are further restricted.
Noticeably, in Example~\ref{ex:tsex}, if $\ptp p$ chooses, e.g.\ $x =
6$ then $\ptp q$ cannot choose a value for $z$.
\end{EX}
Possibly, TS can be regained by rearranging some predicates. In
particular, we can ``lift'' a predicate to a previous interaction
node.
%
%
%
For instance, in Example~\ref{ex:tsex}, one could lift the predicate
$\exists z.x > z > 6$ (adapted from the last interaction)
to the first interaction's predicate.

Without loss of generality, we assume that only one variable is
introduced at the nodes where TS is violated.
Also, we first consider TS violations occurring in interactions and recursive definitions.
Amending violations arising in branching and recursive calls is similar but
complicates the presentation. Hence, for the sake of simplicity, such violations are
considered in \S~\ref{sub:brarec}.

\subsection{Lifting algorithm}
%
%
We formalise the lifting algorithm.
First, we give a function telling us whether a \emph{node} violates
TS.
\begin{definition}[$\TSnodeE$]
  Given $T \in \trees$, $\TSnode{T}{n}$ holds iff $n \in T$,
 and $\TS{T'}$ holds where $T'$ is the assertion tree consisting of the path
  $\uppathPed{n}{T}$ where the children of $n$ (if any) are replaced by nodes
  with label $\End$.
  In addition, we assume that $\TSnodeE$ holds for nodes with label $\TO{s}{r}$.
\end{definition}

We can now define a function that returns a set of nodes violating TS
such that all the previous nodes in the tree do not violate TS.
\begin{definition}[$\TSProblemE$]\label{def:tsbad} The function $\TSProblemE
  : \trees \to \nodes$ is defined as follows:
  \[
  \TSProblem{T} \mmdef
  \left\{n \in T \st 
       \TSnode{T}{n} \text{ is false, and } \TSnode{T}{n'} \ \text{ is true for all }\ n' \in \uppathPed{\parentPed{T}{n}}{T}
  \right\}
  \]
\end{definition}
For instance, in Example~\ref{ex:tsex}, we have that
$\TSProblem{T_{\text{ex}\ref{ex:tsex}}}$ is the singleton
$\{n_{\text{ex}\ref{ex:tsex}}\}$ where $T_{\text{ex}\ref{ex:tsex}} =
\tree{\GA_{\text{ex}\ref{ex:tsex}}}$ and $n_{\text{ex}\ref{ex:tsex}}$
is the node corresponding to the last interaction of
$\GA_{\text{ex}\ref{ex:tsex}}$.
%
%

Once an \emph{interaction} node $n \in \TSProblem T$ is chosen, we rearrange its predicate
as two sub-predicates such that the first one constraints only the variable introduced at $n$, and 
the second one involves other variables (which have been introduced previously in $T$). 
\begin{definition}[$\rewriteE$]\label{def:rewrite}
  Let $\rewriteE : \predicates \times \variables \rightarrow
  \predicates \times \predicates$ be defined as follows:
  \[
  \rewrite{\psi}{v} = 
    \left( \phi(v) , \psi'(\VEC{w}) \right)
  \]
  where $\psi(\VEC{w}) \iff \phi(v) \land \psi'(\VEC{w})$.
\end{definition}
Note that $\rewriteE$ is a non-deterministic total function as $\phi(v)$ could simply be
$\truek$.
The application of $\rewriteE$ to Example~\ref{ex:tsex} yields
$\rewrite{\assert{n_{\text{ex}\ref{ex:tsex}}}}{\carry{n_{\text{ex}\ref{ex:tsex}}}}
= \left( z > 6 , x > z \land y \neq z \right)$.

\begin{REM}
  For a tree $T \in \trees$ and $n \in \TSProblem{T}$ such that
  $\rewrite{\assertPed{T}{n}}{v} = (\phi,\psi')$, we may have
  $\predPed{T}{n} \not\ENTAILSb \exists v . \phi$. For instance, if the
  predicate defined on $v$ alone is not satisfiable, e.g.,
  $\phi=v<7\land v>7$.
 In this case the algorithm is not applicable.
\end{REM}


%
%
We can define a relation among predicates $\psi$ and $\phi$ in a
context $\psi'$ to identify the problematic part of an assertion in an
interaction node.
\begin{definition}[Conflict]\label{def:conflict}
  The predicate $\psi \in \Psi$ \conflict on $\VEC v \subseteq
  \variables$ with $\phi$ \emph{in} $\psi'$ iff
  \[
  \psi' \ENTAILSb \exists \VEC v . \phi
  \quad \text{and} \quad
  \psi' \not\ENTAILSb \exists \VEC v. (\phi \land \psi )
  \]
\end{definition}
Using Def.~\ref{def:conflict} and $\predPed{T}{n}$ (cf.\
\S~\ref{sec:HS}), we define
  {\small\[
  \partition{T}{n}{\phi}{\psi} \mmdef \{ 
  \psi' \st \psi \iff \psi' \land \psi'' \; \text{ and } \; \psi' \text{ \conflict on } \VARS{n} \text{ with } \phi \land \psi'' \emph{ in } \predPed{T}{n}
  \}
  \]}
  which returns a set of problematic predicates.
  Considering again Example~\ref{ex:tsex}, the application of
  $\partitionE$ yields $
  \partition{T_{\text{ex}\ref{ex:tsex}}}{n_{\text{ex}\ref{ex:tsex}}}{z
    > 6}{x > z \land y \neq z} = \{x > z \} $
  since $y \neq z$ allows to choose a suitable value for $z$.
  
  %
%
  The next definition formalises the construction of a new assertion
  tree which possibly regains TS, given a node and an assertion to be
  ``lifted'' (i.e.\ a ``problematic'' predicate).
\begin{definition}[$\buildE$]\label{def:build}
  The function $\build{T}{n}{\psi}$ returns
  \begin{itemize}
  \item $\hat{T} \in \trees$, if we can construct $\hat{T}$
    isomorphic to $T$ except that, each node
    $n' \in \uppathPed{\parentPed{T}{n}}{T}$ such that
     $\lab{n'} = \values s r {\VEC u}{\theta}$ and $\VEC{u} \cap \VARS{\psi} \neq \varnothing$,
     is replaced by a node $\hat{n}$ with label
      \[
      \values s r {\VEC u} {\theta \ \land \ \forall \VEC{x} . \exists \VEC{y}. \psi} \quad \text{such that} \quad
      \theta \land  \forall \VEC{x} . \exists \VEC{y}. \psi
      \text{ is satisfiable}
      \]
      where
      \begin{itemize}
      \item $\VEC{x} \subseteq \VARS{\psi} \setminus \knows{\ptp s}
        T$ are introduced in a node in $\uppathPed{n'}{T}$
      \item $\VEC{y} \subseteq \VARS{\psi}$ are introduced in a node in
        the subtree rooted at $n'$
      \end{itemize}
      and there is no $n' \in \uppathPed{\parentPed{T}{n}}{T}$ such that 
      $\lab{n'} = \recursion{\VEC e}{\VEC v}{\psi} $
      and $\VEC{v} \cap \VARS{\psi} \neq \varnothing$.
  \item $\nada$ otherwise.
  \end{itemize}
\end{definition}
\begin{REM}
  In the definition of $\buildE$, we assume that if either $\VEC x$ or $\VEC y$ is empty,
  the corresponding unnecessary quantifier is removed.
  Recall that global assertions are closed (cf.
  \S~\ref{sec:preliminaries}). Therefore all the variables in
  $\VARS{\psi}$ are taken into account in the construction of the new
  assertion tree.
\end{REM}

In Example~\ref{ex:tsex}, we would invoke
$\build{T_{\text{ex}\ref{ex:tsex}}}{n_{\text{ex}\ref{ex:tsex}}}{z > 6
  \land x > z }$ which returns a new assertion tree. The new tree can
be transformed into a global assertion isomorphic to
$\GA_{\text{ex}\ref{ex:tsex}}$ with line 1 updated to:
$\values{p}{q}{x}{x < 10 \land \exists z.x > z > 6}$.

The function $\TSresolveE : \trees \times \nodes \rightarrow \trees \cup \nada$ either solves a TS problem $n$ or returns $\nada$.
\begin{definition}[$\TSresolveE$]\label{def:resolve} 
  Given $T \in \trees$ and $n \in \TSProblem{T}$, we define
  \[
  \TSresolve{T}{n} = 
  \begin{cases}
    \build{T}{n}{\phi \land \psi'}, &
    \text{if } \lab{n} = \iota \text { and } (\phi, \psi) = \rewrite{\assertPed{T}{n}}{\carryPed{T}{n}}\text{ and there is }
    \\
    & \psi' \in  \partition{T}{n}{\phi}{\psi} \text{ s.t. } \build{T}{n}{\phi \land \psi'} \neq \nada
    \\[.5pc]
    \build{T}{n}{\psi[\VEC{e} / \VEC{v}]}, & \text{if } \lab{n} = \recursion{\VEC e}{\VEC v}{\psi}
    \\[.5pc]
    \nada, & \text{otherwise}
  \end{cases}
  \]
\end{definition}
The second case of Definition~\ref{def:resolve} handles TS violations in
recursive definitions. The problem is similar to the interaction case,
but in this case, the values assigned to the recursion parameters are known (i.e., $\VEC e$).
It may be possible to lift the recursion invariant, where we replace the recursion parameters
by the corresponding initialisation vector.
Example~\ref{ex:ts3} illustrates this case.
\begin{EX}\label{ex:ts3}
  For the global assertion $\GA_{\text{ex}\ref{ex:ts3}}$ given below, 
   $\TS{\GA_{\text{ex}\ref{ex:ts3}}}$
  does not hold because $\truek \not\ENTAILSb (x > y > 6)$. 
  \[
  \begin{array}{llll}
    \GA_{\text{ex}\ref{ex:ts3}}& = &\values{p}{q}{x}{\truek}.\\
    && \recursion{8}{y}{x > y > 6}. \GA' \\
  \end{array}
  \]
 However, using the initialisation parameters, we can
  lift $x > 8 > 6$, i.e., the original predicate where we replaced $y$
  by $8$, to the interaction preceding the recursion.  TS now holds in
  the new global assertion (assuming that $\TS{\GA'}$ holds as well).
\end{EX}
\begin{REM}
  In Example~\ref{ex:ts3}, if we had only lifted $x > y > 6$, as in the 
  interaction case, it would not have solved the TS
  problem. Indeed, the predicate of the first interaction would have
  become $\exists y. x > y >6$ which does not exclude values for
  $x$ which are incompatible with the invariant (e.g., $x = 8$).
\end{REM}


The overall lifting procedure is given. It relies on a repeated application of $\TSresolveE$ until either the assertion tree validates TS or the function fails to solve the problem. In the latter case, the function returns the most improved version of the tree and the node at which it failed.
\begin{definition}[$\liftPredE$]\label{def:lift}
  $\liftPredE$ is defined as follows, given a global assertion $\GA$. 
  \[
  \liftPred{\GA} = 
  \begin{cases}
    \GA, & \text{if } \TS{\GA} \\
    \liftPred{\TSresolve{\tree{\GA}}{n}}, & \text{if there is } n \in \TSProblem{\tree{\GA}} \text{ s.t. } \TSresolve{\tree{\GA}}{n} \neq \nada \\
    \fail{\GA}{n}, & \text{otherwise}
  \end{cases}
  \]
\end{definition}

\subsection{Applying $\liftPredE$ to branching and recursion}\label{sub:brarec}
\paragraph{Branching.}
According to Def.~\ref{def:ts}, TS fails on branching nodes only when
\emph{all} the branches are not satisfiable.  The underlying idea
being that the architect may want to design their choreography in such
a way that a branch cannot be taken when some variables have a
particular value.

Therefore, the architect should be involved in the resolution of the problem.
Two options are possible; either the disjunction of all the predicates
found in the branches is lifted, or one of the branches predicate is
lifted. Arguably, the latter may also prohibit the other branches to
be chosen, as shown in Example~\ref{ex:ts2}.

\begin{EX}\label{ex:ts2}
  As an illustration, we consider the following assertion:
{\small
  \[
  \begin{array}{llll}
    \GA_{\text{ex}\ref{ex:ts2}}& = & \multicolumn{2}{l}{\values{p}{q}{x}{\truek}.} \\
    && \branchingM{p}{q} &\{ v > 5\} \; l_1 : \GA_1\\
    &&                   &\{ v < 5\} \; l_2 : \GA_2
  \end{array}
  \]
}

  Assuming that $\TS{\GA_1}$ and $\TS{\GA_2}$ hold, we have that
  $\TS{\GA_{\text{ex}\ref{ex:ts2}}}$ does not hold because $\truek \not\ENTAILSb (v > 5 \lor v
  < 5)$. It is obvious that if $v = 5$ no branch may be selected.
\end{EX} 
Let's call $\hat{n}$ the node corresponding to the branching in
the second line of $\GA_{\text{ex}\ref{ex:ts2}}$. Depending on the intention of the architect
the problem could be fixed by one of these invocations to $\buildE$
(where, in both cases, superfluous quantifiers are removed).
\begin{itemize}
  \item $\build{\tree{\GA_{\text{ex}\ref{ex:ts2}}}}{\hat{n}}{v > 5 \lor v < 5}$ replaces the predicate in the first line by $\truek \land (v > 5 \lor v < 5)$
  \item $\build{\tree{\GA_{\text{ex}\ref{ex:ts2}}}}{\hat{n}}{v < 5}$ replaces the predicate in the first line by $\truek \land (v < 5)$.
  \end{itemize}
  Both solutions solve the TS problem, however the second one prevents
  the first branch to be ever taken.

  Given an assertion tree $T$ and a branching node\footnote{We also
    assume that TS is not violated in $\uppathPed{\parentPed{T}{n}}{T}$ as in
    Def.~\ref{def:tsbad}.} $n \in T$ such that TS does not hold. One
  can invoke $\build{T}{n}{\psi}$ where $\psi$ is either the
  disjunction of all the branching predicates or one of the branches
  predicate.  If the function does not return $\nada$, then the TS
  problem is solved. Notice that we do not have to use neither
  $\rewriteE$ or $\partitionE$ to solve problems in branching.

\paragraph{Recursion.}
We have seen that when a TS violation is detected in a recursion
definition, lifting may be applied.
However, lifting a predicate involving a recursion parameter $v$ would require
to strengthen the invariant where $v$ is introduced.
This is quite dangerous, therefore the lifting algorithm does not apply in this
case.
In fact, for recursive definition and calls, Def.~\ref{def:ts}
requires $\psi \ENTAILSb \psi'[\VEC e / \VEC v]$, where $\psi'$ is the
recursion invariant and $\psi$ is the conjunction of the previous
predicates.
Hence, lifting a predicate involving a recursion parameter may
strengthen the invariant, and possibly create a new problem in a
corresponding recursive call.
Moreover, notice that, in recursive calls, $\gsatE$
(Def.~\ref{def:ts}) requires that $\psi \land \psi' \ENTAILSb
\psi'[\VEC e / \VEC v]$; namely, strenghtening $\psi'$ would
automatically strenghten $\psi'[\VEC e / \VEC v]$ and therefore leave
the TS problem unsolved.


On the other hand, TS problems can be solved when they occur in
recursive calls. In fact, let a TS problem appear at a node $n \in
T$ such that $\lab{n} = \reccall{\VEC e}$ and let the invariant of the
definition of $\typevar$ being $\psi(\VEC{v})$, then if the invocation
of $\build{T}{n}{\psi[\VEC{e} / \VEC{v}]}$ succeeds, the problem is
solved.

\newcommand{\JPR}{Generator}
\newcommand{\JPQ}{Server}
\newcommand{\JPP}{Player}
In order to give a more complex example of the application of $\liftPredE$, with TS problems in recursive calls, we
consider the following example.
%
%
\begin{EX}\label{ex:ts4}
  Consider the global assertion below
 \[ {\small
  \begin{array}{lllllll}
    \GA_{\text{ex}\ref{ex:ts4}}& = & \multicolumn{3}{l}{\values{\JPR}{\JPQ}{n}{n > 0}.} \\
       &   & \multicolumn{3}{l}{\values{\JPP}{\JPQ}{x}{\truek}.} \\
       &   & \multicolumn{3}{l}{\recursion{x}{r}{ r > 0}.} \\
          &&& \branchingM{\JPQ}{\JPP} &\{ r > n \} \; \textsf{less} :&  \values{\JPP}{\JPQ}{y}{\truek}. \reccall{y}\\
          &&&                         &\{ r < n \} \; \textsf{greater}: & \values{\JPP}{\JPQ}{z}{\truek}. \reccall{z}\\
          &&&                         &\{ r = n \}\; \textsf{win}:& \End
  \end{array}
  } 
\]
modelling a small game where a $\ptp{\JPP}$ has to guess an integer
$n$, following the hints given by a $\ptp{\JPQ}$.
The number is fixed by a $\ptp{\JPR}$.
Each time $\ptp{\JPP}$ sends $\ptp{\JPQ}$ a number, $\ptp{\JPQ}$ says
whether $n$ is less or greater than that number.
\end{EX}
Let $T_{\text{ex}\ref{ex:ts4}}$ be the tree generated from
$\tree{\GA_{\text{ex}\ref{ex:ts4}}}$.  There is a TS problem at the
node corresponding to the recursive definition, indeed if $x \leq 0$, the
invariant is not respected.  After the first loop of
$\liftPred{T_{\text{ex}\ref{ex:ts4}}}$, the predicate $x > 0$ is added
in the second interaction.  Then, the algorithm loops two more times
to solve the problems appearing before the recursive calls.  It adds
$y > 0$ and $z > 0$ in the interaction of the \emph{less} and
\emph{greater} branch, respectively. The global assertion now
validates temporal satisfiability.

\subsection{Properties of $\liftPredE$}

\newcommand{\converge}[1]{#1 \! \! \! \downarrow}

Similarly to the algorithms of \S~\ref{sec:HS}, $\liftPredE$ does not modify the structure of the tree and preserves the properties of the initial assertion.
\begin{PRO}[Underlying Type Structure - $\liftPredE$]\label{pro:struthree}
  Let $\GA$ be a global assertion. If $\liftPred{\GA}$ returns $\GA'$ then
  $erase(\GA)=erase(\GA')$.\footnote{See Section~\ref{def:erase} for the definition of $erase$.}
\end{PRO}

\begin{proof}[Proof sketch]
The proof is by induction on the structure of $\GA$, similarly to the one of Propostion~\ref{prop:str}.
\qedhere
\end{proof}
Also, $\liftPredE$ does not introduce new HS or TS problems.
\begin{PRO}[Properties Preservation - $\liftPredE$]
  Assume $\liftPred{\GA} = \GA'$.
  If $\HSProblem(\GA) = \varnothing$ then $\HSProblem(\GA') = \varnothing$,
  and 
  if $\TSProblem{\GA} = \varnothing$ then $\TSProblem{\GA'} = \varnothing$.
\end{PRO}

\begin{proof}[Proof sketch]
The preservation of HS follows from the fact that all the variables
which are not known to a participant are quantified (either universally or existentially)
in the modified predicates.
The proof of TS preservation follows trivially from the first case of Def.\ref{def:lift}.
\qedhere
\end{proof}

In addition, we have that $\liftPredE$ preserves the domain of
possible values for each variable from the initial assertion.
\begin{PRO}[Assertion predicates]\label{prop:asspred} If $\liftPred{\GA} = \GA'$ then for all $n\in
  \tree{\GA}$ such that $n$ is a leaf, and its corresponding node
  $n'\in \tree{\GA'}$ (cf. Proposition~\ref{pro:struthree})
  \[\predPed{T}{n}\iff \predPed{T'}{n'}\]
\end{PRO}

\begin{proof}[Proof sketch]
The proof follows from the observation that predicates are only duplicated in the tree,
i.e.\ the lifting algorithm does not add any new constraints in the conjunction
of the predicates found on the path from the root to a leaf.
\qedhere
\end{proof}

Finally, 
Proposition~\ref{pro:tsres} establishes an intermediate result for the
correctness of $\liftPredE$. It says that a successful invocation of
$\TSresolveE$ on a node removes the problem at that node.
\begin{PRO}[Correctness - $\TSresolveE$]\label{pro:tsres}
  Let $T$ be an assertion tree, and $N = \TSProblem{T}$.
  For each $n \in N $ such that $\TSresolve{T}{n} \neq \nada$, then  $n \notin \TSProblem{\TSresolve{T}{n}}$.
\end{PRO}

\begin{proof}[Proof sketch]
We sketch the key part of the proof, i.e.\ the proof of the correctness of $\buildE$ for interaction nodes.

Let $T$ be an assertion tree with a node $n$ such that $n \in \TSProblem{T}$,
and $\lab{n} = \values s r v {\phi \land \beta \land \gamma}$ such that $\beta$ \conflict on $\VARS{n}$ with $\phi \land \gamma$ in $\predPed{T}{n}$.
Then $\phi \land \gamma$ is the predicate to be lifted.
Assume $\hat{T} = \build{T}{n}{\phi \land \beta}$.

By Def.\ref{def:build}, we have that, for suitable
$\VEC x_1 , \VEC y_1 \ldots \VEC x_k, \VEC y_k$,
\begin{align}\label{eq:predcon}
\predPed{\hat{T}}{n}  & =  
\predPed{T}{n} \land
\forall \VEC x_1 .\exists \VEC y_1  .(\phi \land \beta) \sigma_1 \land \ldots \land
\forall \VEC x_k .\exists \VEC y_k .(\phi \land \beta) \sigma_k \\
& \iff   \forall \VEC x_1 \ldots \VEC x_k .
\predPed{T}{n} \land
\exists \VEC y_1  .(\phi \land \beta) \sigma_1 \land \ldots \land
\exists \VEC y_k .(\phi \land \beta) \sigma_k 
\end{align}
Where we assume $k$ substitutions $\sigma_i$ such that the variables bound
by $\forall \VEC x_i .\exists \VEC y_i$ in $\phi \land \beta$ are pairwise
distinct.
We have that a quantified version of $\phi \land \beta$ is added $k$ times
in the assertion tree, above $n$.

Note that there must be a $i$ such that
$\exists \VEC y_i .(\phi \land \beta) \sigma_i \iff \exists v . (\phi \land \beta)$.
Indeed, the variables which are quantified \emph{existentially} are the ones that
($i$) appear in $\phi \land \beta$, and ($ii$) are fixed \emph{below} in tree.
Therefore, the predicate which is added in the last node before $n$ must quantify
existentially $v$, only. If there were another variable to be quantified existentially
then it would not be the last node to be updated.

By Def.\ref{def:build}, we also know that every
$\exists \VEC y_i  .(\phi \land \beta) \sigma_i$
is satisfiable.

By the definition of conflict (Def.\ref{def:conflict}), we have that
$\predPed{T}{n} \ENTAILSb \exists v. (\phi \land \gamma)$ and
$\predPed{T}{n} \not\ENTAILSb \exists v. (\phi \land \beta)$ (hence, $\predPed{T}{n}$ is satisfiable).
Therefore, by weakening, we have that
\begin{equation}\label{eq:conf1}
  \predPed{\hat{T}}{n} \ENTAILSb \exists v. (\phi \land \gamma)
\end{equation}
By~\eqref{eq:predcon}, we have that
\begin{equation}\label{eq:predhat}
  \predPed{\hat{T}}{n} \ENTAILSb \exists v . (\phi \land \beta)
\end{equation}
since $\exists v . (\phi \land \beta)$ (modulo renaming) is one
of the conjuncts of $\predPed{\hat{T}}{n}$.

TS must hold for $n$, which implies that $n \not\in \TSProblem{\hat{T}}$ and
$\TSnode{{\hat{T}}}{n}$ holds, i.e.\
\[ 
  \predPed{\hat{T}}{n} \ENTAILSb \exists v. (\phi \land \beta \land \gamma)
\] 

Otherwise, that would imply that
\[ 
  \predPed{\hat{T}}{n} \land \forall v. (\neg\phi \lor \neg\beta \lor \neg\gamma)
\] 
which is in contradiction with~\eqref{eq:conf1} ($\phi$ and $\gamma$)
and~\eqref{eq:predhat} ($\beta$).
\qedhere
\end{proof}

Finally, we can say that, if a repeated application of lifting succeeds, the global assertion which is returned satisfies temporal satisfiability.
\newcommand{\distance}[1]{| #1 |}
\begin{theorem}[Correctness - $\liftPredE$]\label{pro:lift}
$ \text{If } \liftPred{\GA} = \GA' \text{ then } \TSProblem{\GA'} = \varnothing$.
\end{theorem}

\begin{proof}[Proof sketch]
The proof is by induction on the number of problematic
nodes and the minimum depth of these nodes in the tree. It relies on Proposition~\ref{pro:tsres},
i.e.\ the fact that $\TSresolve{T}{n}$ either solves the
problem at $n$ or fails.

Let $T = \tree{\GA}$ and $N$ be the set of nodes in $T$ which violates TS.
We write $\distance{n}$ for the depth of $n$ in $T$ (with $\distance{\treeroot{T}} = 0$).
\begin{enumerate}
\item If $N = \varnothing$, then $T$ is TS.
\item If $N \neq \varnothing$, let $n \in \TSProblem{T} \subseteq N$, after an invocation
  to $\TSresolve{T}{n}$, we have
  \begin{enumerate}
  \item If $\distance{n} > 1$ then either
    \begin{enumerate}
    \item $N := N \setminus \{n\}$, i.e.\ the node is simply removed
      from the set of problematic nodes,
    \item \label{it:boom} $N := N \cup N' \setminus \{n\}$
      with $\forall n'_i \in N' . \, \distance{n_i'} < \distance{n}$,
      i.e.\ the problem at $n$ is solved but other problematic nodes,
      \emph{above} $n$ in $T$, are added,
      or,
    \item the algorithm fails on $n$
    \end{enumerate}
  \item If $\distance{n} \leq 1$ then either $N := N \setminus \{n\}$, or 
    the algorithm fails. In fact, once the algorithm reaches a problem located at a child of the root, then it either fails or solves the problem. Indeed, there cannot be a TS problem at the root node unless the predicate is unsatisfiable (see Def.\ref{def:ts}), in which case, the algorithm fails.
  \end{enumerate}
  Note that selecting $n \in \TSProblem{T}$ implies that the depth of $n$ is smaller or equal to
the depth of the nodes in $N$.
\end{enumerate}
It can be shown by induction that the algorithm terminates either with $\TSProblem{T} = \varnothing$, or a failure.

Regarding step~\ref{it:boom}, note that the algorithm cannot loop on a problematic
node indefinitely. Indeed, the number of (sub)predicates available for lifting is
finite and, by Def.\ref{def:conflict}, the algorithm moves only the
predicates from which the problem originates, e.g.\ an equivalent constraint cannot be
lifted twice.
\qedhere
\end{proof}


\section{A methodology for amending choreographies}\label{sec:methodo}
%
%
%
The algorithms $\methodoneE$, $\methodtwoE$, and $\liftPredE$ in
\S~\ref{sec:HS} and \S~\ref{sec:TS} can be used to support a
methodology for amending contracts in choreographies.
The methodology mainly consists of the following steps: $(i)$ the
architect design a choreography $\widehat \GA$,
$(ii)$ the architect is notified if there are any HS or TS problems in $\widehat \GA$,
$(iii)$ using $\methodoneE$ and $\methodtwoE$ solutions may be offered for HS problems,
while $\liftPredE$ can be used to offer solutions and/or hints on how to solve TS problems;
$(iv)$ the architect picks one of the solutions offered in $(iii)$. Steps $(ii)$ to
$(iv)$ are repeated until all the problems have been solved. 
We sketch our methodology using the following global assertion:
  \[
  \begin{array}{llllll}
  \widehat\GA \quad = \quad
  {\recursion{10}{v}{v>0}.}\\
  \qquad \qquad \qquad {\values{Alice}{Bob}{v_1}{v \geq v_1}.}\\ 
  \qquad \qquad \qquad {\values{Bob}{Carol}{v_2}{v_2>v_1}.} \\
  \qquad \qquad \qquad {\values{Carol}{Alice}{v_3}{v_3>v_1}.}\\
  \qquad \qquad \qquad {\values{Carol}{Bob}{v_4}{v_4>v}.}\\
  \qquad \qquad \qquad \branchingM{Alice}{Bob}  \{ \truek \} \; \textsf{cont}:  \typevar\ENCan{v_1}, \\
  \qquad  \qquad  \qquad\ \  \qquad \qquad \qquad \{ \truek \} \; \textsf{finish}:  \values{Alice}{Bob}{v_5}{v_1 < v_5 < v_3 - 2}
\end{array}
\]
which extends the global assertion in Example~\ref{ex:3}.

%
First, $\widehat\GA$ is inspected by history sensitivity and temporal satisfiability
checkers, such as the ones implemented in~\cite{LT10}.
If there are any HS problems, the $\methodoneE$ and $\methodtwoE$
algorithms are used, while $\liftPredE$ is used for TS problems.
This allows the architect to detect all the problems and
consider the ones for which (at least) one of the algorithms is
applicable.

We assume here that the architect focuses on HS problems first.
In $\widehat\GA$ there are two HS problems, both of
them can be solved automatically, and the methodology will return that
\begin{enumerate}
\item At line 4, $v_1$ is not known by $\ptp{Carol}$; the problem is solvable by either
  \begin{itemize}
  \item replacing $v_3 > v_1$ by $v_3 > v_2$ (algorithm $\methodoneE$)
    at line 4, or
  \item by revealing $v_1$ to $\ptp{Carol}$ (algorithm $\methodtwoE$);
    in this case, line 3 becomes
    \[ \values{Bob}{Carol}{v_2 \; u_1}{v_2>v_1 \land u_1 = v_1} \]
    and the assertion at line 4 becomes $v_3 > u_1$.
  \end{itemize}
\item At line 5, $v$ is not known by $\ptp{Carol}$; the problem is
  solvable by revealing the value of $v$ to $\ptp{Carol}$ (algorithm
  $\methodtwoE$) in which case line 3 becomes
  \[\values{Bob}{Carol}{v_2 \; u_2}{v_2>v_1 \land u_2
    = v}\]
  and the assertion at line 5 becomes $v_4 > u_2$.
\end{enumerate}
In the \emph{propagation case} (i.e., $\methodtwoE$), the methodology
gives the architect information on which participants the value of a
variable may be disclosed to.
Indeed, as discussed in Remark~\ref{rmk:secrecy}, it may not be
appropriate to use the suggested solution.
Therefore, the actual adoption of the proposed solutions should be
left to the architect.
In addition, the order in which problems are tacked is also left to
the architect (e.g., the same variable may be involved in several problems and
solving one of them may automatically fix the others).
%
%
%
Assuming that $\methodoneE$ is used to solve the first problem and
$\methodtwoE$ to solve the second, the first five lines of the new
global assertion are those in Example~\ref{ex:5} and HS is fixed.
%
%

%
Now HS is satisfied in $\widehat\GA$, but TS problems are
still there.
In case a TS problem cannot be solved automatically, additional
information can be returned: $(a)$ at which node the problem occurred,
$(b)$ which variables or recursion parameters are posing problems (i.e.\ using
$\partitionE$ and $\buildE$), and $(c)$ where liftings are not possible
(i.e. when $\buildE$ fails to add a satisfiable predicate to a node).
For $\widehat\GA$ there are two TS problems which are dealt with
sequentially. The methodology would report that
\begin{enumerate}
\item At line 6, $v_1$ does not satify the invariant $v > 0$. This
  can be solved by lifting $v_1 > 0$ (i.e.\ the invariant where $v$ is replaced
  by the actual parameter $v_1$) to the interaction at line 2,
  which would yield the new assertion $v \geq v_1 \land v_1 > 0$.
\item At line 7, there might be no value for $v_5$ such that $v_1 <
  v_5 < v_3 - 2$. The assertion is \emph{in conflict} (cf.
  Def.~\ref{def:conflict}) with the previous predicates; this problem
  cannot be solved since lifting would add the following predicates in
  line 2 and 4, respectively.
  \begin{itemize}
  \item $\exists v_3, v_5 . v_1 < v_5 < v_3 - 2$ which is indeed
    satisfiable, but remarkably does not constraint $v_1$ more than
    the initial predicate.
  \item $\forall v_1 . \exists v_5 . v_1 < v_5 < v_3 - 2$ which is not satisfiable, therefore the algorithm fails.
  \end{itemize}
\end{enumerate}
The failure of $\liftPredE$ is due to the fact that $v_5$ is
constrained by $v_1$ and $v_3$ which are fixed by two different
participants.
They would have to somehow interact in order to guarantee that there
exists a value for $v_5$, this cannot be done automatically.
Notice that in this case the methodology tells the architect that
$v_5$, fixed by $\ptp{Alice}$, is constrained by $v_1$ and $v_3$ which
are fixed by $\ptp{Alice}$ and $\ptp{Carol}$, respectively.
Our methodology can also suggest that the node introducing $v_3,$ or
(the part of) the assertion over $v_3$ may be the source of the problem
since $v_3$ is the only variable not known by $\ptp{Alice}$.

\begin{REM}\label{rmk:interferece}
  The application of an algorithm could compromise the application of
  another one due to some ``interference'' effect that may arise.
  For instance, applying strengthening ($\methodoneE$) could spoil the application of
  lifting ($\liftPredE$) and vice versa (cf. \S~\ref{sec:conc} for an intuitive
  explanation).
\end{REM}


\section{Conclusions}\label{sec:conc}

In this paper, we investigated the problem of designing consistent
assertions. We focused on two consistency criteria from~\cite{bhty10}:
history sensitivity and temporal satisfiability. We proposed and
compared three algorithms ($\Phi_1$, $\Phi_2$, and $\Phi_3$) to amend
global assertions. Since each algorithm is applicable only in certain
circumstances, we proposed a methodology that supports
the architect when violations are not automatically amendable.

On the theoretical side, the algorithms $\Phi_1$, $\Phi_2$, and
$\Phi_3$ address the general problem of guaranteeing the
satisfiability of predicates when: (1) the parts of the system have a
different perspective/knowledge of the available information (in the
case of history sensitivity), and (2) the constraints are introduced
progressively (in the case of temporal satisfiability). The proposed
solutions can be adapted and used, for instance, to amend processes
(rather than types), orchestrations (rather than choreographies, when
we want to check for local constraints), e.g., expressed in formalisms
as CC-Pi~\cite{BuscemiMontanari}, a language for distributed processes
with constraints.
%
Interestingly, temporal satisfiability is similar to the feasibility
property in~\cite{Apt:1987:AFD:41625.41642} requiring that any initial
segment of a computation must be possibly extended to a full computation
to prevent ``a scheduler from
`painting itself into a corner' with no possible continuation''.
A promising future development is to investigate more general
accounts of satisfiability which is applicable to different
application scenarios.

In scope of future work, we will study the ``interference'' issues of
the three algorithms (see Remark~\ref{rmk:interferece}) so to refine
our methodology and use them more effectively.
We conjecture, for instance, that conflicts between $\methodoneE$ and
$\liftPredE$ appear \emph{only} when the variable introduced where an
HS problem is solved by $\methodoneE$ is also involved in a TS
problem. 
More precisely, let $v$ be introduced at a node $n$ having an HS
problem.
If $\methodoneE$ is used to solved such problem the constraint at $n$
will be strengthened.
Now, if a node $n'$ --further down than $n$ in the tree-- has a TS
problem with a conflict involving $v$, the predicate at $n$ will be
updated (i.e.\ strengthened) by $\liftPredE$.
Therefore, the predicate at $n$ would be strengthened by each algorithm in an
independent way.
This may render the predicate at $n$ unsatisfiable.

We will also study the applicability of our methodology in more
realistic cases in order to assess the quality of the solutions offered by our
algorithms.

We plan to implement our algorithms and support for the methodology
by integrating it in the tool introduced in~\cite{LT10}.

%
%
%
%
%
%
%
%


{
\bibliographystyle{eptcs}
\bibliography{session,socrw}

\begin{thebibliography}{1}
\providecommand{\bibitemdeclare}[2]{}
\providecommand{\urlprefix}{Available at }
\providecommand{\url}[1]{\texttt{#1}}
\providecommand{\href}[2]{\texttt{#2}}
\providecommand{\urlalt}[2]{\href{#1}{#2}}
\providecommand{\doi}[1]{doi:\urlalt{http://dx.doi.org/#1}{#1}}
\providecommand{\bibinfo}[2]{#2}

\bibitemdeclare{article}{Apt:1987:AFD:41625.41642}
\bibitem{Apt:1987:AFD:41625.41642}
\bibinfo{author}{Krzysztof~R. Apt}, \bibinfo{author}{Nissim Francez} \&
  \bibinfo{author}{Shmuel Katz} (\bibinfo{year}{1988}):
  \emph{\bibinfo{title}{Appraising fairness in languages for distributed
  programming}}.
\newblock {\sl \bibinfo{journal}{Distributed Computing}} \bibinfo{volume}{2},
  pp. \bibinfo{pages}{226--241}.

\bibitemdeclare{inproceedings}{BettiniCDLDY08LONG}
\bibitem{BettiniCDLDY08LONG}
\bibinfo{author}{Lorenzo Bettini}, \bibinfo{author}{Mario Coppo},
  \bibinfo{author}{Loris D'Antoni}, \bibinfo{author}{Marco~De Luca},
  \bibinfo{author}{Mariangiola Dezani-Ciancaglini} \& \bibinfo{author}{Nobuko
  Yoshida} (\bibinfo{year}{2008}): \emph{\bibinfo{title}{Global Progress in
  Dynamically Interleaved Multiparty Sessions}}.
\newblock In \bibinfo{editor}{Franck van Breugel} \& \bibinfo{editor}{Marsha
  Chechik}, editors: {\sl \bibinfo{booktitle}{CONCUR}}, {\sl
  \bibinfo{series}{Lecture Notes in Computer Science}} \bibinfo{volume}{5201},
  \bibinfo{publisher}{Springer}, pp. \bibinfo{pages}{418--433}.
\newblock \urlprefix\url{http://dx.doi.org/10.1007/978-3-540-85361-9_33}.

\bibitemdeclare{inproceedings}{bhty10}
\bibitem{bhty10}
\bibinfo{author}{Laura Bocchi}, \bibinfo{author}{Kohei Honda},
  \bibinfo{author}{Emilio Tuosto} \& \bibinfo{author}{Nobuko Yoshida}
  (\bibinfo{year}{2010}): \emph{\bibinfo{title}{A Theory of Design-by-Contract
  for Distributed Multiparty Interactions}}.
\newblock In \bibinfo{editor}{Paul Gastin} \& \bibinfo{editor}{Fran\c{c}ois
  Laroussinie}, editors: {\sl \bibinfo{booktitle}{CONCUR}}, {\sl
  \bibinfo{series}{Lecture Notes in Computer Science}} \bibinfo{volume}{6269},
  \bibinfo{publisher}{Springer}, pp. \bibinfo{pages}{162--176},
  \doi{10.1007/978-3-642-15375-4\_12}.
\newblock \urlprefix\url{http://dx.doi.org/10.1007/978-3-642-15375-4_12}.

\bibitemdeclare{inproceedings}{BuscemiMontanari}
\bibitem{BuscemiMontanari}
\bibinfo{author}{Maria~Grazia Buscemi} \& \bibinfo{author}{Ugo Montanari}
  (\bibinfo{year}{2007}): \emph{\bibinfo{title}{CC-Pi: a constraint-based
  language for specifying service level agreements}}.
\newblock In: {\sl \bibinfo{booktitle}{Proceedings of the 16th European
  conference on Programming}}, \bibinfo{series}{ESOP'07},
  \bibinfo{publisher}{Springer-Verlag}, \bibinfo{address}{Berlin, Heidelberg},
  pp. \bibinfo{pages}{18--32}, \doi{10.1007/978-3-540-71316-6}.
\newblock
  \urlprefix\url{http://portal.acm.org/citation.cfm?id=1762174.1762179}.

\bibitemdeclare{inproceedings}{carbone.honda.yoshida:esop07}
\bibitem{carbone.honda.yoshida:esop07}
\bibinfo{author}{Marco Carbone}, \bibinfo{author}{Kohei Honda} \&
  \bibinfo{author}{Nobuko Yoshida} (\bibinfo{year}{2007}):
  \emph{\bibinfo{title}{Structured Communication-Centred Programming for Web
  Services}}.
\newblock In: {\sl \bibinfo{booktitle}{{19th International Conference on
  Concurrency Theory (Concur'08)}}}, \bibinfo{publisher}{Springer}, pp.
  \bibinfo{pages}{2--17}, \doi{10.1007/978-3-540-71316-6}.
\newblock
  \urlprefix\url{http://www.eecs.qmul.ac.uk/~carbonem/cdlpaper/esop2007.pdf}.

\bibitemdeclare{inproceedings}{mps}
\bibitem{mps}
\bibinfo{author}{Kohei Honda}, \bibinfo{author}{Nobuko Yoshida} \&
  \bibinfo{author}{Marco Carbone} (\bibinfo{year}{2008}):
  \emph{\bibinfo{title}{Multiparty asynchronous session types}}.
\newblock In: {\sl \bibinfo{booktitle}{POPL}}, pp. \bibinfo{pages}{273--284},
  \doi{10.1145/1328438.1328472}.
\newblock \urlprefix\url{http://doi.acm.org/10.1145/1328438.1328472}.

\bibitemdeclare{inproceedings}{LT10}
\bibitem{LT10}
\bibinfo{author}{Julien Lange} \& \bibinfo{author}{Emilio Tuosto}
  (\bibinfo{year}{2010}): \emph{\bibinfo{title}{A Modular Toolkit for Theories
  of Distributed Interactions}}.
\newblock In: {\sl \bibinfo{booktitle}{PLACES}}.
\newblock \bibinfo{note}{To appear}.

\bibitemdeclare{book}{scoop}
\bibitem{scoop}
\bibinfo{author}{Bertrand Meyer} (\bibinfo{year}{1997}):
  \emph{\bibinfo{title}{Object-Oriented Software Construction (Chapter 31)}}.
\newblock \bibinfo{publisher}{Prentice Hall}.

\end{thebibliography}
}

\end{document}